\documentclass[12pt]{article}

\usepackage{amsmath}
\usepackage{amssymb}
\usepackage{appendix}
\usepackage{amsthm}
\usepackage{graphicx}
\usepackage{setspace}
\usepackage{enumitem}
\usepackage{hyperref}

\usepackage[utf8]{inputenc}
\usepackage[english]{babel}
\usepackage{setspace}

\usepackage[margin=0.8in]{geometry}

\usepackage{etoolbox}

\numberwithin{equation}{section}

\theoremstyle{definition}

\newtheorem{theorem}{Theorem}[section]
\newtheorem{corollary}[theorem]{Corollary}
\newtheorem{lemma}[theorem]{Lemma}
\newtheorem{proposition}[theorem]{Proposition}
\newtheorem{example}{Example}[section]
\newtheorem{definition}{Definition}
\newtheorem{assume}{Assumption}

\AtEndEnvironment{example}{\null\hfill\qedsymbol}

\theoremstyle{remark}

\begin{document}

\author{Hyungbin Park\thanks{hyungbin@cims.nyu.edu, hyungbin2015@gmail.com}\\ \\
{\normalsize Courant Institute of Mathematical Sciences,}\\ {\normalsize New York University, New York, NY, USA}\\ \\
{\normalsize 23 September 2015}}
\title{Ross Recovery with Recurrent and Transient Processes\footnote{The author is grateful to Jonathan Goodman and Srinivasa Varadhan for technical insights. 
The author also appreciates two unknown referees for their constructive feedback.} \thanks{First version: 19 September 2014}}

\date{}
\maketitle

\begin{abstract}
Recently, Ross \cite{Ross13} argued 
that it is possible to recover an objective measure from a risk-neutral measure.
His model assumes that there is a finite-state Markov process $X_{t}$ that drives the economy
in discrete time $t\in\mathbb{N}.$
Many authors extended his model to a continuous-time setting with a Markov diffusion process 
$X_{t}$ with state space $\mathbb{R}.$ Unfortunately, the continuous-time
model fails to recover an objective measure from a risk-neutral measure in general. 
We determine under which information recovery is possible 
in the continuous-time model.
It was proven that if $X_{t}$ is {\em recurrent} under the objective measure, then recovery is possible.
In this article,
when $X_{t}$ is {\em transient} under the objective measure,
we investigate what information is sufficient to recover.\newline

\noindent Keywords: Ross recovery, Markovian pricing operators, recurrence, transience
\end{abstract}

\section{Introduction}
\label{sec:intro}

Quantitative finance theory involves two related probability measures: a risk-neutral measure and an objective measure.
The risk-neutral measure determines the
prices of assets and options in a financial market. The risk-neutral measure is distinct from the 
objective measure, which describes the actual stochastic dynamics of markets. The conventional 
belief is that one cannot determine an objective measure by observing a risk-neutral measure. The 
best known example capturing this belief is the Black-Scholes model, which says that the drift 
of a stock under a risk-neutral measure is independent of the drift of the stock under an
objective measure.

Recently, Ross \cite{Ross13} questioned this belief and argued that it is possible to recover an objective
measure from a risk-neutral measure under some circumstances. His model assumes that there 
is an underlying process $X_{t}$ that drives the entire economy with a finite number of states on 
discrete time $t\in\mathbb{N}.$ This result can be of great interest to finance researchers and 
investors, and thus it is highly valuable to extend the Ross model to a continuous-time setting, 
which is practical and useful in finance.

In this paper, we investigate the possibility of recovering in a continuous-time
setting $t\in\mathbb{R}$ with a time-homogeneous Markov diffusion process $X_{t}$ with state 
space
$\mathbb{R}.$ In this setting, the risk-neutral measure contains some information about an
objective measure. In general, however, the model unfortunately fails to recover an objective measure from a risk-neutral measure.

A key idea of recovery theory is that the reciprocal of the pricing kernel
is expressed in the form
$e^{\beta t}\,\phi(X_{t})$
for some constant $\beta$ and positive function $\phi(\cdot).$ 
For example, in the {\em consumption-based capital asset model} in \cite{Breeden79} and \cite{Karatzas98}, the pricing kernel is expressed in the above form.
The basis of recovery theory is 
finding $\beta$ and $\phi(\cdot).$ Thus, we obtain the pricing kernel and the 
relationship between the objective measure and the risk-neutral measure.

We will see that $\beta$ and $\phi(\cdot)$ satisfy the second-order differential
equation 
\begin{equation} \label{eqn:DE}
\frac{1}{2}\sigma^2(x){\phi ''(x)}+k(x)\phi '(x)
-r(x)\phi (x) =-\beta \,\phi (x)\,.
\end{equation}
Thus, recovery theory is transformed into a problem of finding a particular solution 
pair $(\beta,\phi)$ of this particular differential equation with $\phi(\cdot)>0.$ If such a solution 
pair were unique, then we could successfully recover the objective measure. Unfortunately, this 
approach categorically fails to achieve recovery because such a solution pair is never unique.

Many authors have extended the Ross model to a continuous-time setting and have also
confronted the non-uniqueness problem. To overcome the non-uniqueness problem, all 
authors assumed more conditions onto their models so that the differential equation \eqref{eqn:DE} has a unique solution pair 
satisfying the conditions.

Carr and Yu \cite{Carr12} introduced the notion of Long's discovery of the numeraire portfolio
to extend the Ross model to a continuous-time setting. 
They assumed Long's portfolio depends on time $t$ 
and the underlying process $X_{t},$ and then they derived the above differential equation \eqref{eqn:DE}. Carr and Yu also 
assumed that the process $X_{t}$ is a time-homogeneous Markov diffusion on a {\em bounded} 
interval with regular boundaries at both endpoints. They also implicitly assumed that $\phi(\cdot)$ is in 
$L^{2}(w)$ for some measure $w$ to apply the regular Sturm-Liouville theory, thereby obtaining 
a unique solution pair satisfying these conditions.
Dubynskiy and Goldstein \cite{Dubynskiy13}
explored Markov diffusion models with reflecting boundary conditions.

Walden \cite{Walden13} extended the results of Carr and Yu to the case that $X_{t}$ is an 
{\em unbounded} process. Walden proved that recovery is possible if the process $X_{t}$ is 
{\em recurrent} under the objective measure. In addition, he showed that when recovery is possible 
in the unbounded case, approximate recovery is possible from observing option prices on a 
bounded subinterval.

Qin and Linetsky \cite{Qin14a} proved that recovery is possible if $X_{t}$ is recurrent and the pricing 
kernel admits a Hansen-Scheinkman decomposition. They also 
showed that the Ross recovery has a close connection with Roger's potential approach to the 
pricing kernel. Borovicka, Hansen and Scheinkman \cite{Borovicka14} showed that the recovery is possible if the 
process $X_{t}$ is {\em stochastically stable} under the objective measure. They also discussed 
applications of the recovery theory to finance and economics.

The papers of Borovicka, Hansen and Scheinkman \cite{Borovicka14}, Qin and Linetsky
\cite{Qin14a} and Walden \cite{Walden13} assumed a common condition on $X_{t}$. Specifically, $X_{t}$ is {\em recurrent} under the objective measure. The mathematical rationale for 
this condition is to overcome the non-uniqueness problem of the differential equation \eqref{eqn:DE}. Indeed, 
if existent, there is a unique solution pair $(\beta,\phi)$ of the equation \eqref{eqn:DE} satisfying this condition and we will review this condition in 
Section \ref{sec:recurrent_recovery}.

In this article, we investigate the possibility of recovery when the process $X_{t}$ is {\em
transient} under the objective measure. We explore in this case what information is sufficient to recover. One of the main contributions is that if $\beta$ is known and if $X_{t}
$ is {\em non-attracted} to the left (or right) boundary under the objective measure, then recovery is possible.
To achieve this, we establish a graphical understanding of recovery theory. 
This topic is discussed in Section \ref{sec:recurrence_and_transience} and \ref{sec:recovery_theory}.
In Section \ref{sec:applications}, two examples of recovery theory are explored:
the Cox-Ingersoll-Ross (CIR) interest rate model and the Black-Scholes stock model.
Section \ref{sec:conclusion} summarizes this article.

\section{Markovian pricing operators}
\label{sec:Markovian_pricing}

A financial market is defined as a probability space
$(\Omega,\mathcal{F},\mathbb{P})$ having a one-dimensional Brownian motion $B_{t}$
with the filtration $\mathcal{F}=(\mathcal{F}_{t})_{t=0}^{\infty}$ generated by $B_{t}$.  
All the processes
in this article are assumed to be adapted to the filtration $\mathcal{F}$. 
$\mathbb{P}$ is the objective measure of this market.
We assume that there are a state variable $X_t$ and a positive numeraire $G_t$ in the market.

Let $\mathbb{Q}$ be an equivalent measure on the market $(\Omega,\mathcal{F},\mathbb{P})$ such that each risky asset discounted by the numeraire $G_t$ is a martingale under measure $\mathbb{Q}.$
It is customary that this measure $\mathbb{Q}$ is referred to as a risk-neutral measure when $G_t$ is a money market account. 
In this article, however, for any given positive numeraire $G_t,$ we say $\mathbb{Q}$ is a risk-neutral measure with respect to $G_t.$  
Set the Radon-Nikodym derivative by 
$\Sigma_{t}=\left.\frac{d \mathbb{Q}}{d \mathbb{P}}
\right|_{\mathcal{F}_{t}},$
which is known to be a martingale process 
on $(\Omega,\mathcal{F},\mathbb{P})$ for $0<t<T.$ Using the martingale representation theorem, we can write in the stochastic differential equation form 
$d\Sigma_{t}=-\rho_{t}\Sigma_{t}\, dB_{t}$ for some $\rho_{t}$.
 It is well-known that $W_{t}$ defined by
\begin{equation}\label{eqn:Girsanov}
dW_{t}:=\rho_{t}dt+dB_{t}
\end{equation}
 is a Brownian motion under $\mathbb{Q}.$
We define {\em the reciprocal of the pricing kernel}  by $L_{t}=G_t/\Sigma_{t}.$

\begin{assume} \label{assume:X}
The state variable $X_{t}$ is a
time-homogeneous Markov diffusion process satisfying 
$$dX_{t}=b(X_{t})\,dt+\sigma(X_{t})\,dW_{t}\,,\;X_{0}=\xi\;.$$
The range of $X_t$ is an open interval $I=(c,d)$ with $-\infty\leq c<d\leq\infty.$
$b(\cdot)$ and $\sigma(\cdot)$ are continuously differentiable on $I$
and that $\sigma(x)>0$ for $x\in(c,d).$
\end{assume}
\noindent It is implicitly assumed that both endpoints are unattainable because the range of the process is an open interval.

\begin{assume} \label{assume:interest_rate}
The dynamics of the numeraire $G_t$ is determined by $X_{t}.$ More precisely, $G_t$ follows
$$\frac{dG_t}{G_t}=(r(X_t)+v^2(X_t))\,dt+v(X_t)\,dW_t\;,\; G_0=1\;.$$
Assume that $r$ and $v$ are continuously differentiable on $I$ and 
$$\exp{\left(-\frac{1}{2}\int_{0}^t v^2(X_s)\,ds-\int_0^t v(X_s)\,dW_s\right)}$$
is a martingale.
\end{assume}
\noindent We assume that we can extract theses four functions $b(\cdot),\sigma(\cdot),r(\cdot)$ and $v(\cdot)$ from market prices data, thus they are assumed to be known ex ante. The above martingale assumption is to define a new measure by using the Girsanov theorem, for example, in the proof of Theorem \ref{thm:criterion_martingality}.
It is noteworthy that if 
there is a money market account with interest rate, denoted by $r_t,$ in the market, then $r_t$ is equal to $r(X_t)$ because $e^{\int_0^tr_s\,ds}\cdot G_t^{-1}$ is a martingale under $\mathbb{Q}.$

\begin{assume} \label{assume:transition_indep}
Assume that (the reciprocal of) the pricing kernel $L_t$ is {\em transition independent} in the sense that 
there are a positive function $\phi\in C^{2}(I)$ and a real number $\beta$ 
such that
\begin{equation} \label{eqn:RN_X}
L_{t}=e^{\beta t}\,\phi(X_{t})\,\phi^{-1}(\xi)\; .
\end{equation}
In this case, we say $(\beta,\phi)$ is a
{\em principal pair} of the market.
\end{assume}
\noindent The basis of recovery theory is finding the principal pair $(\beta,\phi)$ and then obtaining the
objective measure $\mathbb{P}$ by setting the Radon-Nikodym derivative
$$\left.\frac{d\mathbb{P}}{d\mathbb{Q}}\right|_{\mathcal{F}_{t}}=\Sigma_t^{-1}=e^{\beta t}\,\phi(X_{t})\,\phi^{-1}(\xi)\,G_t^{-1}\;.$$

One important aspect in implementing the recovery approach is to decide how to choose state variables $X_t.$
Many processes can serve as a state variable and the choice of a state variable depends on the purpose of use.
One way is the short interest rate $r_{t}.$
Investors interested in the price of bonds want to find the dynamics of $r_t$ under an objective measure.
Plenty of examples with interest rate state variables can be found in \cite{Qin14b}.
Another way is a stock market index process such as the Dow Jones Industrial Average and Standard $\&$ Poor’s (S$\&$P) 500.
Refer to \cite{Audrino} for an empirical analysis of recovery theory with the state variable S$\&$P 500.

\section{Transformed measures}
\label{sec:transformed}
We investigate how recovery theory is transformed into a problem of a differential equation.
Applying the Ito formula to the definition of $L_t,$ we know
$$dL_{t}=(r(X_t)+v^2(X_t)+v(X_t)\rho_t)\,L_{t}\,dt +(v(X_t)+\rho_{t})L_{t}\, dW_{t}\,.$$
From \eqref{eqn:RN_X}, we also have
$$dL_{t}=\left(\beta+\frac{1}{2}(\sigma^{2}\phi''\phi^{-1})(X_{t})+(b\phi'\phi^{-1})(X_{t})\right)L_{t}
\,dt+(\sigma\phi'\phi^{-1})(X_{t})\,L_{t}\,dW_{t}\;.$$
By comparing these two equations, we obtain
\begin{equation}\label{eqn:rho}
\left\{ \quad
\begin{aligned}
&\frac{1}{2}\sigma^2{\phi ''}+(b-v\sigma)\phi'
-r\phi=-\beta \,\phi\,,\\
&\rho_{t}=(\sigma\phi'\phi^{-1}-v)(X_{t})\;.
\end{aligned}\right.
\end{equation}
For convenience, set $k(x):=(b-v\sigma)(x).$
Using notation $\mathcal{L}$ defined by
$\mathcal{L}\phi(x)=\frac{1}{2}\sigma^2(x){\phi ''(x)}+k(x)\phi '(x)
-r(x)\phi (x),$
we have the following theorem.
\begin{theorem}
Let $(\beta,\phi)$ be
the principal pair of the market. Then $(\beta,\phi)$
satisfies 
$\mathcal{L}\phi=-\beta\phi.$
\end{theorem} \noindent
In other words, if $(\lambda,h)$
is a solution pair of $\mathcal{L}h=-\lambda h$
with $h>0,$
then
$(\lambda,h)$ is a candidate pair for the principal pair of $X_{t}.$

We are interested in a solution pair $(\lambda,h)$ of
$\mathcal{L}h=-\lambda h$ with positive function $h.$
There are two possibilities.
\begin{itemize}[noitemsep,nolistsep]
\item[\textnormal{(i)}] there is no positive solution $h$ for any $\lambda\in\mathbb{R}$, or
\item[\textnormal{(ii)}] there exists a number $\overline{\beta}$ such that it
has two linearly independent positive solutions for $\lambda<\overline{\beta},$ has no positive solution for
$\lambda>\overline{\beta}$ and has one or two linearly independent solutions for $\lambda=\overline{\beta}.$  
\end{itemize}
Refer to page 146 and 149 in \cite{Pinsky}.
In this article, we implicitly assumed the second case by Assumption \ref{assume:transition_indep}.

It is easily checked that
$$e^{\lambda t}\,h(X_{t})\,h^{-1}(\xi)\,G_t^{-1}$$
is a local martingale under $\mathbb{Q}.$
When this is a martingale, one can {\em attempt} to recover the objective measure $\mathbb{P}$ by setting this as a Radon-Nikodym derivative.

\begin{definition}
Let $(\lambda,h)$ be a solution pair of
$\mathcal{L}h=-\lambda h$ with positive function $h.$
Suppose that $e^{\lambda t}\,h(X_{t})\,h^{-1}(\xi)\,G_t^{-1}$ is a martingale.
A measure obtained
from the risk-neutral measure $\mathbb{Q}$
by the Radon-Nikodym derivative 
$$\left.\frac{\,d\,\cdot\,}{d\mathbb{Q}}\right|_{\mathcal{F}_{t}}=e^{\lambda t}\,h(X_{t})\,h^{-1}(\xi)\,G_t^{-1}$$
is called
{\em the transformed measure} with respect to the pair $(\lambda,h).$
\end{definition}
\noindent
Clearly, the transformed measure with respect to the principal pair is the objective measure 
$\mathbb{P}.$
We have the following proposition by \eqref{eqn:Girsanov} and \eqref{eqn:rho}.
\begin{proposition} \label{prop:dynamics_under_P}
A process $B_{t}^{h}$ defined by
$dB_{t}^{h}=-(\sigma h'h^{-1}-v)(X_{t})\,dt+dW_{t}$
is a Brownian motion under the transformed measure with respect to $(\lambda,h).$
Furthermore,
$X_{t}$ follows
\begin{equation}\label{eqn:X_under_P}
\begin{aligned}
dX_{t}&=(b-v\sigma+\sigma^{2}h'h^{-1})(X_{t})\,dt+\sigma(X_{t})\,dB_{t}^{h}\\
&=(k+\sigma^{2}h'h^{-1})(X_{t})\,dt+\sigma(X_{t})\,dB_{t}^{h}\;.
\end{aligned}
\end{equation}
\end{proposition}

\noindent Occasionally, we use the notation $B_t$ instead of $B_t^h$ without ambiguity.

Even when $e^{\lambda t}\,h(X_{t})\,h^{-1}(\xi)\,G_t^{-1}$ is not a martingale, we can consider the diffusion process corresponding to \eqref{eqn:X_under_P}.

\begin{definition}
The diffusion process $X_t$ defined by 
$$dX_t=(k+\sigma^2 h'h^{-1})(X_t)\,dt+\sigma(X_t)\,dB_t$$
is called the diffusion process {\em induced by} $(\lambda,h).$
\end{definition}

\section{Recurrence and transience}
\label{sec:recurrence_and_transience}
\subsection{Mathematical preliminaries}
\label{sec:math_preli}
We establish the mathematical preliminaries for recurrent and transient 
processes. The contents of this section are indebted to \cite{Durrett96}, \cite{Karatzas91}
and \cite{Karlin81}.
Consider the diffusion process induced by $(\lambda,h):$
$$dX_t=\mu(X_t)\,dt+\sigma(X_t)\,dB_t\,,\;X_0=\xi\,,$$
where $\mu(\cdot)=(k+\sigma^{2}h'h^{-1})(\cdot).$ 
A measure $S$ defined by
$$dS(x):
=e^{-\int_{\xi}^{x}\frac{2\mu(y)}{\sigma^{2}(y)}\,dy}\,dx
=\frac{h^2(\xi)}{h^2(x)}\,e^{-\int_{\xi}^{x}\frac{2k(y)}{\sigma^{2}(y)}\,dy}\,dx$$
is called {\em the scale measure} of the process with respect to the pair $(\lambda,h).$
\begin{definition}
The left boundary $c$ is {\em attracting} if $S((c,\xi])<\infty$
and {\em non-attracting} otherwise.
Similarly, we say the right boundary $d$ is {\em attracting} if $S([\xi,d))<\infty$
and {\em non-attracting} otherwise.
\end{definition}
\noindent
Define a stopping time $\tau$ by the following way.
Let $(c_{n})_{n=1}^{\infty}, (d_{n})_{n=1}^{\infty}$ be strictly monotone sequences with limits $c$ and $d,$ respectively. Set $\tau_{n}:=\inf\,\{\,t>0\,|\, X_{t} \notin (c_{n},d_{n})\, \}$ and $\tau:=\lim_{n\rightarrow \infty}\tau_{n}.$

\begin{proposition}
The left boundary $c$ is non-attracting if and only if $$\textnormal{Prob}\left(\lim_{t\rightarrow\tau}X_{t}=c \right)=0\;.$$
It is similar to the right boundary $d.$
\end{proposition}

\begin{proposition}
$X_{t}$ is recurrent if and only if both boundary points $c$ and $d$ are non-attracting.
\end{proposition}

\subsection{Graphical understanding}
\label{sec:graphical_understaing}
We establish a graphical understanding of recovery theory.
Recall that $X_0=\xi$ in Assumption \ref{assume:X}.
The purpose of this section is to understand the graphs shown in Figure \ref{fig:candidate} and \ref{fig:admissible}. 
A solution $h$ of a second-order differential equation is uniquely determined by the initial value 
and the initial velocity.
By normalizing, we may assume $h(\xi)=1$ such that a solution is determined by $h'(\xi).$
In this section, we describe
the relationship between $h'(\xi)$ and the recovery of $X_{t}$ with respect to $h.$ 
Occasionally, we use the terminology without ambiguity:
the transformed measure with respect to {\em tuple}
$(\lambda,h'(\xi))$ means the transformed measure with respect to {\em pair} $(\lambda,h).$
The two terms tuple and pair will be used to distinguish between these meanings.

\begin{definition} 
We say $(\lambda,h'(\xi))\in \mathbb{R}^{2}$ is a candidate tuple or we say
$(\lambda,h)$ is a candidate pair
if
$(\lambda,h)$ is a solution pair of $\mathcal{L}h=-\lambda h$ with
$h(\cdot)>0$ and $h(\xi)=1.$ 
Denote the set of the candidate tuples by $\mathcal{C}.$
$$\mathcal{C}:=\left.\left\{(\lambda,h'(\xi))\in\mathbb{R}^{2}\right| \mathcal{L}h=-\lambda 
h,\, h>0,\, h(\xi)=1 \right\}.$$
\end{definition} \noindent

We investigate graphical properties
of $\mathcal{C}.$
Let $\overline{\beta}$ be the maximum value of the first coordinate of elements of 
$\mathcal{C},$ that is, $\overline{\beta}:=\max\{\; \lambda \;|\;(\lambda,z)\in\mathcal{C} \;\}\;.$
The maximum $\overline{\beta}$ is achieved as we discussed in Section \ref{sec:transformed}.
For any $\lambda$ with $\lambda\leq \overline{\beta},$
we set
$$M_{\lambda}:=\sup_{(\lambda,z)\in\mathcal{C}}z\,,\;\;
m_{\lambda}:=\inf_{(\lambda,z)\in\mathcal{C}}z\;.$$
\begin{proposition} \label{prop:slice}
Let $\lambda\leq\overline{\beta}.$
For any $z$ with $m_{\lambda}\leq z\leq M_{\lambda},$ the tuple $(\lambda,z)$ is in 
$\mathcal{C}.$
\end{proposition}
\noindent
Therefore, the supremum and infimum are, in fact, the maximum and minimum, respectively. 
Furthermore, the $\lambda$-slide of $\mathcal{C}$ is a connected and compact set.
See Appendix \ref{app:pf_slice} for proof.

\begin{theorem} \label{thm:main_thm}
The diffusion process 
induced by tuple $(\lambda,M_{\lambda})$
is non-attracted to the left boundary.
For $z$ with $m_{\lambda}\leq z<M_{\lambda},$
the diffusion process induced by tuple $(\lambda,z),$
is attracted to the left boundary.
\end{theorem}
\noindent
Similarly, the diffusion process induced by tuple $(\lambda,m_{\lambda})$ is non-attracted to the right boundary.
For $z$ with $m_{\lambda}< z\leq M_{\lambda},$
the diffusion process induced by tuple $(\lambda,z)$
is attracted to the right boundary.
Therefore, for $z$ with $m_{\lambda}< z< M_{\lambda},$
the diffusion process induced by tuple $(\lambda,z)$ is attracted to both boundaries.
For proof, see Appendix \ref{app:pf_main_thm}.

\begin{proposition} \label{prop:monotone}
$M_{\lambda}$ is a strictly decreasing function of $\lambda$ and $m_{\lambda}$ is a strictly increasing function of $\lambda$ for $\lambda\leq\overline{\beta}.$
\end{proposition}
\noindent
See Appendix \ref{app:pf_monotone} for proof. \newline

\begin{corollary} \label{cor:two_cases}
For $\overline{\beta},$ there are two possibilities:
\begin{enumerate}
\item[\textnormal{(i)}] there is a unique number $\overline{z}$ such that $(\overline{\beta},\overline{z})$ is in $\mathcal{C}.$
In this case, the tuple $(\overline{\beta},\overline{z})$ is the unique tuple in $\mathcal{C}$ such 
that the induced diffusion process is recurrent.
 
\item[\textnormal{(ii)}] there is an infinite number of $z$'s such that $(\overline{\beta},z)$ is in $\mathcal{C}.$
In this case, for any such tuple $(\overline{\beta},z)$ in $\mathcal{C},$ the induced diffusion process 
$X_{t}$ is transient.
\end{enumerate}
\end{corollary} 
\noindent Refer to Section \ref{sec:applications} for examples of (i) and also see Appendix \ref{app:example_2nd}
for an example of (ii).
In Figure \ref{fig:candidate}, the left graph is the case of (i) and the right graph is the case of (ii).

\begin{figure}[ht]
\centering
\begin{minipage}[b]{0.45\linewidth}
\includegraphics[scale=0.65]{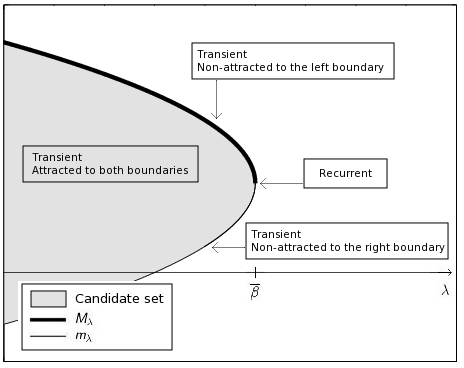}
%\caption{Happy Smiley}
%\label{fig:minipage1}
\end{minipage}
\quad
\begin{minipage}[b]{0.45\linewidth}
\includegraphics[scale=0.65]{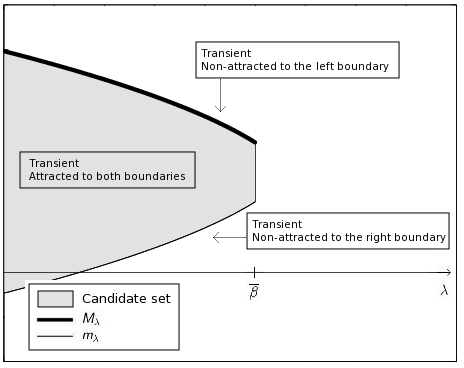}
%\caption{Sad Smiley}
%\label{fig:minipage2}
\end{minipage}
\caption{Candidate sets}
\label{fig:candidate}
\end{figure}

We now explore a particular subset of $\mathcal{C}.$
In general, for a candidate pair $(\lambda,h),$
$e^{\lambda t}\,h(X_{t})\,G_t^{-1}$
is a local martingale.
We are interested in pairs that induce martingales.

\begin{definition} Let $(\lambda,h)$ be a candidate pair.
We say $(\lambda,h'(\xi))$ is an admissible tuple or we say $(\lambda,h)$ is an admissible pair if 
$e^{\lambda t}\,h(X_{t})\,G_t^{-1}$ is a martingale.
Denote the set of the admissible tuples by $\mathcal{A}.$
$$\mathcal{A}:=\left.\left\{(\lambda,h'(\xi))\in\mathbb{R}^{2}\right|
(\lambda,h) \text{ is a candidate pair and }
e^{\lambda t}\,h(X_{t})\,G_t^{-1}\text{ is a martingale}
\right\}.$$
\end{definition}

\noindent One of the main purposes of this section is to investigate the graphical properties
of $\mathcal{A}$ with the notion of recurrence and transience. 
It is noteworthy that $\mathcal{C}$ and $\mathcal{A}$ depend on $\xi.$

\begin{proposition} \label{prop:connected_martingality}
Let $\lambda\leq\overline{\beta}$ and let $m_{\lambda}<z< M_{\lambda}.$
If two tuples $(\lambda,m_\lambda)$ and $(\lambda,M_\lambda)$ are in $\mathcal{A},$ then 
$(\lambda,z)$ is in $\mathcal{A}.$
If at least one of $(\lambda,m_\lambda)$ and $(\lambda,M_\lambda)$ is not in $\mathcal{A},$ then 
$(\lambda,z)$ is not in $\mathcal{A}.$
\end{proposition}

\noindent See Appendix \ref{app:pf_prop:connected_martingality} for proof. \newline

\begin{proposition}\label{prop:monotone_martingality}
Let $\delta<\lambda\leq\overline{\beta}.$ If $(\delta,M_\delta)$ is in $\mathcal{A},$ then $(\lambda,M_\lambda)$ is also in $\mathcal{A}.$
Similarly, if $(\delta,m_\delta)$ is in $\mathcal{A},$ then $(\lambda,m_\lambda)$ is also in $\mathcal{A}.$
\end{proposition}
\noindent Refer to Appendix \ref{app:pf_monotone_martingality} for proof.

\begin{proposition}
Consider the case of \textnormal{(i)} in Corollary \ref{cor:two_cases}. For $\overline{\beta},$ suppose that 
there is a unique number $\overline{z}$ such that $(\overline{\beta},\overline{z})$ is in $\mathcal{C}.$
Then  $(\overline{\beta},\overline{z})$ is in $\mathcal{A}.$ 
\end{proposition}
\noindent For proof, see \cite{Pinsky}.
Figure \ref{fig:admissible} 
displays two examples of admissible set and summarizes this section.

\begin{figure}[ht]
\centering
\begin{minipage}[b]{0.45\linewidth}
\includegraphics[scale=0.65]{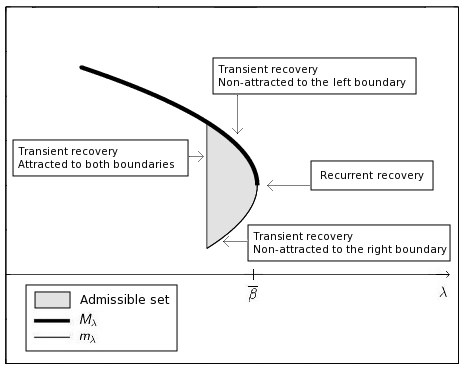}
%\caption{Happy Smiley}
%\label{fig:minipage1}
\end{minipage}
\quad
\begin{minipage}[b]{0.45\linewidth}
\includegraphics[scale=0.65]{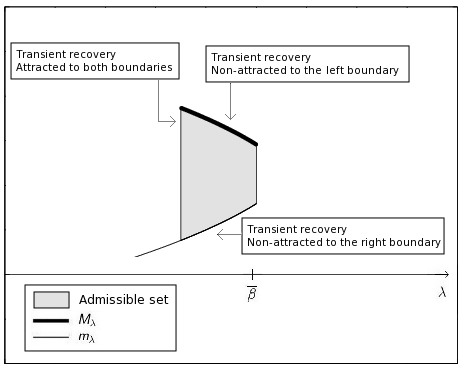}
%\caption{Sad Smiley}
%\label{fig:minipage2}
\end{minipage}
\caption{Admissible sets}
\label{fig:admissible}
\end{figure}

\section{Recurrent and transient recovery}
\label{sec:recovery_theory}

\subsection{Recurrent recovery}
\label{sec:recurrent_recovery}
We review recovery theory
with the assumption that
$X_{t}$ is {\em recurrent} under the objective measure.
With this assumption, we can successfully recover the objective measure from the risk-neutral measure.
This theory is especially useful when the state variable $X_t$ is an interest rate process because the actual dynamics of an interest rate is usually recurrent or mean-reverting in the actual real-world measure.
\begin{proposition} \label{prop:recurrence}
If it exists, there is a unique admissible pair $(\overline{\beta},\overline{\phi})$ of $\mathcal{L}
\overline{\phi}=-\overline{\beta}\,\overline{\phi}$ with
$\overline{\phi}>0$ 
such that $X_{t}$ is recurrent under the transformed measure with respect 
to the pair $(\overline{\beta},\overline{\phi}).$
In this case, we have
$\overline{\beta}=\max\{\; \lambda \;|\;(\lambda,z)\in\mathcal{A} \;\}\,.$
\end{proposition}
\noindent This proposition is easily obtained from Section \ref{sec:graphical_understaing} and gives the following theorem.

\begin{theorem}\label{thm:recurrent_recovery}
\textnormal{(Recurrent recovery)}
Suppose $X_{t}$ is recurrent under the objective measure $\mathbb{P}.$ We can then recover 
the objective measure $\mathbb{P}$ from the risk-neutral measure $\mathbb{Q}.$
\end{theorem}

\subsection{Transient recovery}
\label{sec:transient_recovery}
We encounter several conditions under which we can recover the objective 
measure when $X_t$ is transient under the objective measure.
The state variable $X_t$ is always transient under the transformed measure with respect to an admissible pair $(\lambda,h)$ for any $\lambda<\overline{\beta}.$
Therefore,
without further information,
recovery is impossible when the
process $X_{t}$ is transient under the objective measure.

We assume that we know the value $\beta.$
Despite knowing this value $\beta,$ we cannot achieve recovery in general.
However, the following theorem says that recovery is possible under some circumstances.
The proof is direct from Proposition \ref{prop:connected_martingality}.

\begin{theorem}
\textnormal{(Transient recovery)}
Suppose we know the value $\beta.$
If only one of $(\beta,m_\beta)$ and $(\beta,M_\beta)$ is an admissible tuple, then
we can recover the objective measure $\mathbb{P}$ from the risk-neutral measure $\mathbb{Q}.$
\end{theorem}

\noindent  
When both of $(\beta,m_\beta)$ and $(\beta,M_\beta)$ are admissible tuples,
we cannot uniquely determine the objective measure $\mathbb{P}$ because there is
an infinite number of admissible pairs.

We investigate another way for recovery. 
We will see that there is only one way to recover the objective measure such that $X_{t}$ is non-attracted to the left (or right) boundary.

\begin{proposition}
For any fixed $\lambda$ with $\lambda\leq\overline{\beta},$ if it exists, there is a unique 
admissible pair $(\lambda,h)$ such that $X_{t}$ is non-attracted to the left boundary
under the corresponding transformed measure.
In this case,
$h'(\xi)=M_{\lambda}.$
\end{proposition}
\noindent 
Similarly, 
if it exists, there is a unique 
admissible pair $(\lambda,h)$ such that $X_{t}$ is non-attracted to the {\em right} boundary
under the corresponding transformed measure.
In this case,
$h'(\xi)=m_{\lambda}.$
This proposition is easily obtained from Section \ref{sec:graphical_understaing} and gives the following theorem.

\begin{theorem}\label{thm:transient_recovery}
\textnormal{(Transient recovery)}
Suppose we know the value $\beta.$
If $X_{t}$ is non-attracted to the left (or right) boundary under the 
objective measure $\mathbb{P},$ then we can recover the objective measure $\mathbb{P}$ 
from the risk-neutral measure $\mathbb{Q}.$
\end{theorem}

\noindent If $X_{t}$ is attracted to both boundaries under the objective measure $\mathbb{P},$ 
then both $(\beta,m_\beta)$ and $(\beta,M_\beta)$ are admissible tuples.
Thus there is an infinite number of admissible tuples such that $X_{t}$ is attracted to both boundaries under the corresponding transformed measure. \newline

We now shift our attention to the choice of $\beta.$
When $X_{t}$ is non-attracted to the left (right) 
boundary, to recover the objective measure, we confront a problem of determining the value $\beta.$ 
How can we choose the value?
One way is to use the long-term yield of bonds, which is defined by 
\begin{equation*}
\lim_{t\rightarrow \infty}\left(-\frac{1}{t}\cdot\log \mathbb{E}^{\mathbb{Q}}\left[
G_t^{-1}\right]\right)\;.
\end{equation*}
See \cite{Martin} and \cite{Qin14b} as a reference.

\section{Applications}
\label{sec:applications}

\subsection{Interest rates}
\label{sec:interest_rates}
We investigate the recovery theorem when the state variable $X_t$ is an interest rate process $r_t$ and the numeraire is the money market account:
$$X_t=r_t\;,\;G_t=e^{\int_0^tr_s\,ds}\;.$$
Assume that $r_t$ follows
$dr_t=k(r_t)\,dt+\sigma(r_t)\,dW_t.$
Then 
$$\frac{1}{2}\sigma^2(r)h''(r)+k(r)h'(r)-rh(r)=-\lambda h\;.$$
is the corresponding second-order differential equation.
 
\begin{example}
We explore the CIR model:
$$dr_{t}=a(\theta-r_{t})\,dt+\sigma\sqrt{r_{t}}\,dW_{t}$$
with the Feller condition $2a\theta\geq\sigma^2.$ 
Consider
$$\frac{1}{2}\sigma^{2}rh''(r)+a(\theta-r)h'(r)-rh(r)=\lambda h$$
with $h(r_{0})=1.$
Set $k:=\frac{\sqrt{a^{2}+2\sigma^{2}}-a}{\sigma^{2}}.$
It can be shown that the maximum value $\overline{\beta}$ is equal to $ka\theta.$
For any fixed $\lambda$ with $\lambda\leq ka\theta,$ denote the functions corresponding to tuples $(\lambda,M_\lambda)$ and
$(\lambda,m_\lambda)$ by $h_\lambda(\cdot)$ and $g_\lambda(\cdot),$ respectively.
Then we have $h_\lambda(x)=\psi_\lambda(x)/\psi_\lambda(r_0)$ where
$$\psi_\lambda(x):=e^{-kx}\,K\left(\frac{\lambda-ka\theta}{\sqrt{a^2+2\sigma^2}}\,,\,\frac{2a\theta}{\sigma^2}\,,\,\frac{2x\sqrt{a^2+2\sigma^2}}{\sigma^2}\right)\;.$$
Here, $K(\,\cdot,\,\cdot\,,\cdot\,)$ is the Kummer confluent hypergeometric function.
Refer to \cite{Qin14a}  for more details about $g_\lambda.$
It can be shown that $(\lambda,M_\lambda)$ is an admissible tuple, but $(\lambda,m_\lambda)$ is not.
Thus, the admissible set $\mathcal{A}$ is described by
$$\mathcal{A}=\{\,(\lambda,h_\lambda'(r_0))\in\mathbb{R}^2\,|\,\lambda\leq ka\theta\,\}\;.$$
See \cite{Qin14a}  for more details.

Interest rates are usually recurrent in the real world; therefore, 
we assume that $r_{t}$ is recurrent under the objective measure.
$(ka\theta, e^{-k(r-r_{0})})$
is the only admissible pair that
induces the recurrent $X_t$ under the corresponding transformed measure and under which
the $X_t$ is expressed by
$$dr_{t}=\sqrt{a^{2}+2\sigma^{2}}\left(\frac{a\theta}{\sqrt{a^{2}+2\sigma^{2}}}-r_{t}\right)\,dt+\sigma\sqrt{r_{t}}\,dB_{t}\;.$$
\end{example}

\subsection{Stock prices}
\label{sec:stock_prices}
We investigate the recovery theorem when the state variable $X_t$ is a stock price $S_t$ and 
the dividends of the stock are paid out continuously with
rate $\delta(S_t)\,dt.$
Here $\delta(\cdot)$ is a deterministic function and the case of $\delta(\cdot)=0$ is not excluded. 
Let the numeraire $G_t$ be the wealth process defined by 
$$G_t=S_te^{\int_0^t\delta(S_u)\,du}\;.$$
Assume that the dynamics of $S_t$ is given by
$$dS_t=(r(S_t)-\delta(S_t)+\sigma^2(S_t))S_t\,dt+\sigma(S_t)S_t\,dW_t\;.$$
Then it follows that
$$\frac{dG_t}{G_t}=(r(S_t)+\sigma^2(S_t))\,dt+\sigma(S_t)\,dW_t\;.$$
The corresponding second-order differential equation is
$$\frac{1}{2}\sigma^2(s)s^2h''(s)+(r(s)-\delta(s))sh'(s)-r(s)h(s)=-\lambda h(s)\;.$$
It is noteworthy that that if 
there is a money market account with a constant interest rate $r$ in the market, then $r(S_t)=r.$

\begin{example}
Assume that the state variable is a stock price $S_t$ and
the dividends of the stock are paid out continuously with
rate $\delta\,dt.$
Suppose $S_t$ follows a geometric Brownian motion
$$dS_t=(r-\delta+\sigma^2)S_t\,dt+\sigma S_t\,dW_t\;,\;S_0=1$$
and the numeraire is $G_t=S_te^{\delta t}.$ 
Consider
$$\frac{1}{2}\sigma^2s^2h''(s)+(r-\delta)sh'(s)-rh(s)=-\lambda h(s)\;.$$
We want to find the candidate set and the admissible set.
For $\lambda<\frac{1}{2}\left(\frac{\sigma}{2}-\frac{r-\delta}{\sigma}\right)^2+r,$
these solutions are given by
$h(s)=c\,s^{l_{1}}+(1-c)\,s^{l_{2}}$
for $0\leq c\leq 1$
where
\begin{equation*}
\begin{aligned}
l_{1}=\frac{1}{2}-\frac{r-\delta}{\sigma^2}-\sqrt{\left(\frac{1}{2}-\frac{r-\delta}{\sigma^2}\right)^{2}+\frac{2(r-\lambda)}{\sigma^2}}\,,\\
l_{2}=\frac{1}{2}-\frac{r-\delta}{\sigma^2}+\sqrt{\left(\frac{1}{2}-\frac{r-\delta}{\sigma^2}\right)^{2}+\frac{2(r-\lambda)}{\sigma^2}}\;,
\end{aligned}
\end{equation*}	
thus we have
$m_{\lambda}=l_{1},\,M_{\lambda}=l_{2}.$
For $\lambda=\frac{1}{2}\left(\frac{\sigma}{2}-\frac{r-\delta}{\sigma}\right)^2+r,$
the only positive solution $h$ with $h(1)=1$ is
$h(x)=x^{\frac{1}{2}-\frac{r-\delta}{\sigma^2}}$
and $m_{\lambda}=M_{\lambda}=\frac{1}{2}-\frac{r-\delta}{\sigma^{2}}.$
For $\lambda>\frac{1}{2}\left(\frac{\sigma}{2}-\frac{r-\delta}{\sigma}\right)^2+r,$
there are no positive solutions. 
Thus, the candidate set is obtained. It can be easily shown that the admissible set is equal to the candidate set. 

If $S_{t}$ is recurrent under the objective 
measure, then
the transformed measure with respect to tuple 
$\left(\frac{1}{2}\left(\frac{\sigma}{2}-\frac{r-\delta}{\sigma}\right)^2+r,\frac{1}{2}-\frac{r-\delta}
{\sigma^{2}}\right)$ is the objective measure.
We used Theorem \ref{thm:recurrent_recovery}.
Under this measure, $S_{t}$ follows
$$dS_{t}=\frac{1}{2}\sigma^2 S_t\,dt+\sigma S_t\,dB_{t}\;.$$

We now assume that
$S_{t}$ is non-attracted to $0$ and the 
value $\beta\; \left(\leq \frac{1}{2}\left(\frac{\sigma}{2}-\frac{r-\delta}{\sigma}\right)^2+r\right)$ is known.
Then, by Theorem \ref{thm:transient_recovery}, 
the transformed measure with respect to tuple 
$$(\beta,M_\beta)=\left(\beta\,,\,\frac{1}{2}-\frac{r-\delta}{\sigma^2}+\sqrt{\left(\frac{1}{2}-\frac{r-\delta}{\sigma^2}\right)^{2}+\frac{2(r-\beta)}{\sigma^2}}\right)$$ is the objective measure,
under which $S_{t}$ follows
$$dS_{t}=\left(\frac{\sigma^2}{2}
+\sqrt{\left(\frac{\sigma^2}{2}-(r-\delta)\right)^2+2\sigma^2(r-\beta)}\right)S_t\,dt+\sigma S_t\,dB_{t}\,.$$
As a particular case, if $\beta$ is equal to the long-term yield of bonds
$$\lim_{t\rightarrow \infty}\left(-\frac{1}{t}\cdot\log \mathbb{E}^{\mathbb{Q}}\left[
G_t^{-1}\right]\right)=r\;,$$
then $S_{t}$ follows
$$dS_{t}=\left(\frac{\sigma^2}{2}+\left|r-\delta-\frac{\sigma^2}{2}\right|\right)S_t\,dt+\sigma S_t\,dB_{t}$$
under the objective measure.
The numeraire $G_t$ satisfies $dG_t=\mu G_t\,dt+\sigma G_t\,dB_{t}$
with $$\mu:=\frac{\sigma^2}{2}+\delta+\left|r-\delta-\frac{\sigma^2}{2}\right|\;.$$
Thus, the Sharpe ratio of $G_t$ is
\begin{equation*}
\frac{\mu-r}{\sigma}=\left\{
\begin{aligned}
&\sigma+\frac{2(\delta-r)}{\sigma}&&\quad\textnormal{ if } r<\delta+\frac{\sigma^2}{2}\\
&\quad\quad 0 &&\quad\textnormal{ if } r\geq\delta+\frac{\sigma^2}{2}\;.
\end{aligned}\right.
\end{equation*}	
\end{example}

\section{Conclusion}
\label{sec:conclusion}
This article extended the Ross model to a continuous-time setting model. 
In a continuous-time setting, the risk-neutral measure contains some information about the objective measure. 
Unfortunately, the model fails to recover an objective measure from a risk-neutral measure. 
We discussed several conditions under which the 
recovery of the objective measure from the risk-neutral measure is possible in a continuous-time 
model. When the state variable $X_t$ is recurrent under the objective measure, recovery is possible.

We also determined the type of information that is sufficient for recovery when $X_t$ is transient.
It was shown that when $X_t$ is transient,
recovery is possible if the value $\beta$ is known and only one of $(\beta,m_\beta)$ and $(\beta,M_\beta)$ is an admissible tuple.
We can also recover the objective measure if the value $\beta$ is known and $X_t$ is non-attracted to the left (or right) boundary.

The following extensions for future research are suggested. First, it would be interesting to find sufficient conditions under 
which recovery is possible when $X_t$ is attracted to both boundaries. We could 
not offer such conditions in this article.
Second, it would be valuable to find financially and economically reasonable ways to determine $\beta.$ 
Finally, much work remains 
to be conducted on the implementation and empirical testing of recovery theory in future
research.

\appendices

\section{Proof of Proposition \ref{prop:slice}}
\label{app:pf_slice}

\begin{proposition} \label{prop:ODE}
A solution $h$ of $\mathcal{L}h=-\lambda h$
can be expressed by
$h=uq,$ where $q(x)=e^{-\int_{0}^{x}\frac{k(y)}{\sigma^{2}(y)}\,dy}$ and $u$ is a solution of
$$u''(x)+\left(-\frac{d}{dx}\left(\frac{k(x)}{\sigma^{2}(x)}\right)-\frac{k^{2}(x)}{\sigma^{4}(x)}
+\frac{2(-r(x)+\lambda)}{\sigma^{2}(x)}\right)u(x)=0\;.$$
\end{proposition} \noindent
This can be shown by direct calculation.
For more details, refer to \cite{Walden13}, page 36. \newline

We now prove Proposition \ref{prop:slice}.

\begin{proof}
For convenience, we may assume that $\xi=0$, the left boundary is $-\infty$ and the right boundary is  $\infty.$
Let $h$ be a positive solution of $\mathcal{L}h=-\lambda h$ with $h(0)=1.$
Another solution, which is independent of $h$, is 
$$h(x)\int_{0}^{x}\frac{1}{h^{2}(y)}\,e^{-\int_{0}^{y}\frac{2k(z)}{\sigma^{2}(z)}\,dz}\,dy\;.$$
(It can be obtained by direct calculation).
Let $S$ be the scale measure with respect to the pair $(\lambda,h),$ then 
$$h(x)\int_{0}^{x}\frac{1}{h^{2}(y)}\,e^{-\int_{0}^{y}\frac{2k(z)}{\sigma^{2}(z)}\,dz}\,dy=h(x)\cdot S([0,x))\;.$$
The general solutions of $\mathcal{L}h=-\lambda h$ are expressed by
$h(x)\left(c_{1}+c_{2}\cdot S([0,x))\right).$
Assume that at least one of
$S((-\infty,0])$ and $S([0,\infty))$ is finite. (Otherwise, $h$ is the 
unique positive solution of $\mathcal{L}h=-\lambda h$, so $m_{\beta}=z=M_{\beta},$ in which 
case we have nothing to prove).
We may assume $S((-\infty,0])<\infty.$
Set $B:=S((-\infty,0]).$
The general (normalized to $h_{c}(0)=1$) solution is expressed by
$$h_{c}(x):=h(x)\left((1-c)+\frac{c}{B}\cdot S((-\infty,x])\right)\;,$$
which is a positive function for and only for 
\begin{equation*}
\left\{
\begin{aligned}
0\leq &c\leq 1 &&\text{ if } S([0,\infty))=\infty   \\
-\frac{B}{D}\leq &c\leq 1 &&\text{ if } D:=S([0,\infty))<\infty
\end{aligned}\right.
\end{equation*}
Using $h_{c}'(0)=h'(0)+\frac{c}{B},$ we have that
\begin{equation*}
\left\{
\begin{aligned}
&M_{\lambda}=h'(0)+\frac{1}{B},\,&&m_{\lambda}=h'(0) &&\text{ if } S([0,\infty))=\infty \\
&M_{\lambda}=h'(0)+\frac{1}{B},\,&&m_{\lambda}=h'(0)-\frac{1}{D} &&\text{ if }D=S([0,\infty))<\infty
\end{aligned}\right.
\end{equation*}
Furthermore, $h_{c}'(0)$ can be any value in $[m_{\lambda}, M_{\lambda}].$
Hence, for any $z$ with
$m_{\lambda}\leq z\leq M_{\lambda},$ the tuple $(\lambda,z)$ is in $\mathcal{A}.$ 
This completes the proof.
\end{proof}

\section{Proof of Theorem \ref{thm:main_thm}}
\label{app:pf_main_thm}

\begin{proof}
For convenience, 
we may assume that $\xi=0$, the left boundary is $-\infty$ and the right boundary is  $\infty.$
Suppose that the diffusion process induced by tuple $(\lambda,M_\lambda)$ is attracted to the left boundary.
Denote the scale measure with respect to tuple $(\lambda,M_\lambda)$ by $S.$
We have $B:=S((-\infty,0])<\infty.$
By the proof of Proposition \ref{prop:slice}, we know that $(\lambda,h'(0)+\frac{1}{B})
=(\lambda,M_{\lambda}+\frac{1}{B})$ is in $\mathcal{C}.$ This is a contradiction.

We now show that for $z$ with $m_{\lambda}\leq z<M_{\lambda},$
the diffusion process induced by tuple $(\lambda,z)$
is attracted to the left boundary.
Let $h$ and $g$ be the functions corresponding to tuple $(\lambda,M_\lambda)$ and $(\lambda,z),$ respectively.
Recall $q(x)=e^{-\int_{0}^{x}\frac{k(y)}{\sigma^{2}(y)}\,dy}$ in Proposition \ref{prop:ODE}.
Write $h=uq$ and $g=vq.$
It can be easily checked that $u(0)=v(0)=1$ and $u'(0)>v'(0).$
Set $\Gamma:=\frac{u'}{u}-\frac{v'}{v}.$
By direct calculation, we have
$\Gamma'=-\Gamma^{2}-\frac{2v'}{v}\Gamma.$
Because $\Gamma(0)>0$ and $\Gamma=0$ is an equilibrium point, we know that $\Gamma(x)
\geq 0$ for all
$x.$ 
By differentiating $\Gamma(x)\,e^{\int_{0}^{x}\left(\Gamma(y)+\frac{2v'(y)}{v(y)}\right)\,dy},$
we have
$$\left(\Gamma'(x)+\Gamma^{2}(x)+\frac{2v'(x)}{v(x)}\Gamma(x)\right)e^{\int_{0}^{x}\left(\Gamma(y)
+\frac{2v'(y)}{v(y)}\right)\,dy}=0\;.$$
Thus, $\Gamma(x)\,e^{\int_{0}^{x}\left(\Gamma(y)+\frac{2v'(y)}{v(y)}\right)\,dy}=\Gamma(0),$
which yields
$$v^{-2}(x)=e^{-\int_{0}^{x}\frac{2v'}{v}\,dy}=\frac{\Gamma(x)}{\Gamma(0)}\,e^{\int_{0}^{x}
\Gamma(y)\,dy}\;.$$
We obtain that 
\begin{equation*}
\begin{aligned}
\int_{-\infty}^{0}\frac{1}{g^{2}(y)}\,e^{-\int_{0}^{y}\frac{2k(z)}{\sigma^{2}(z)}\,dz}\,dy\
&=\int_{-\infty}^{0}v^{-2}(y)\,dy \\
&=\frac{1}{\Gamma(0)}\int_{-\infty}^{0}\Gamma(y)\,e^{\int_{0}^{y}\Gamma(z)\,dz}\,dy\\
&=\frac{1}{\Gamma(0)}\left(1-e^{-\int_{-\infty}^{0}\Gamma(y)\,dy}\right) \leq \frac{1}{\Gamma(0)}
<\infty\;.
\end{aligned}
\end{equation*}
Therefore, the diffusion process induced by $(\lambda,z)$ is attracted to $-\infty.$
\end{proof}

\section{Proof of Proposition \ref{prop:monotone}}
\label{app:pf_monotone}

\begin{proof}
Let $\delta<\lambda.$
Let $g$ be a positive function of $\mathcal{L}g=-\delta g$ with $g(0)=1$ and $g'(0)=M_{\delta}.$
We show that there does not exist a positive function $h$ of $\mathcal{L}h=-\lambda h$ with 
$h(0)=1$ and $h'(0)=M_{\delta}.$ This implies that $M_{\delta}>M_{\lambda}.$
Suppose
$h$ is a positive solution of $\mathcal{L}h=-\lambda h$ with $h(0)=1$ and $h'(0)=M_{\delta}.$
Recall $q(x)=e^{-\int_{0}^{x}\frac{k(y)}{\sigma^{2}(y)}\,dy}$ in Proposition \ref{prop:ODE}.
Write $h=uq$ and $g=vq.$
Define $\Gamma=\frac{u'}{u}-\frac{v'}{v}.$
Then
$$\Gamma'=-\Gamma^{2}-\frac{2v'}{v}\Gamma-\frac{2(\lambda-\delta)}{\sigma^{2}}\;.$$
From $\Gamma(0)=0,$ we have that $\Gamma(x)>0$ for $x<0$
because if $\Gamma$ ever gets close to $0,$ then term $-\frac{2(\beta-\lambda)}{\sigma^{2}}$ dominates the right hand side of the equation.
Choose $x_{0}$ with $x_{0}<0.$
For $x<x_{0},$ we have
$$-\frac{2v'(x)}{v(x)}=\frac{\Gamma'(x)}{\Gamma(x)}+\Gamma(x)+\frac{2(\lambda-\delta)}
{\sigma^{2}(x)}\cdot\frac{1}{\Gamma(x)}\;.$$
Integrating from $x_{0}$ to $x,$
$$-2\ln\frac{v(\,x\,)}{v(x_{0})}=\ln \frac{\Gamma(\,x\,)}{\Gamma(x_{0})}+\int_{x_{0}}^{x}
\Gamma(y)\,dy+\int_{x_{0}}^{x}\frac{2(\lambda-\delta)}{\sigma^{2}(y)}\cdot\frac{1}{\Gamma(y)}
\,dy$$
which leads to
$$\frac{v^{2}(x_{0})}{v^{2}(\,x\,)}\leq \frac{\Gamma(\,x\,)}{\Gamma(x_{0})}\,e^{\int_{x_{0}}^{x}
\Gamma(y)\,dy}$$
for $x<x_{0}.$
Thus,
\begin{equation*}
\begin{aligned}
\int_{-\infty}^{x_{0}}\frac{1}{g^{2}(y)}\,e^{-\int_{0}^{y}\frac{2k(z)}{\sigma^{2}(z)}\,dz}\,dy\
&=\int_{-\infty}^{x_{0}}\frac{1}{v^{2}(y)}\,dy\\
&\leq (\text{constant})\cdot\int_{-\infty}^{x_{0}}\Gamma(y)\,e^{\int_{x_{0}}^{y}\Gamma(w)dw}\,dy 
\\
&=(\text{constant})\cdot\left(1-e^{-\int_{-\infty}^{x_{0}}\Gamma(w)dw}\right)\\
&\leq (\text{constant})\\
&< \infty \;.
\end{aligned}
\end{equation*}
This implies that the diffusion process induced by the tuple $(\delta,M_{\delta})$ is attracted to the left boundary, 
which is a contradiction.
\end{proof}

\section{Proof of Proposition \ref{prop:connected_martingality}}
\label{app:pf_prop:connected_martingality}

Suppose that two tuples $(\lambda,m_\lambda)$ and $(\lambda,M_\lambda)$ are in $\mathcal{A}.$
Let $h_m$ and $h_M$ be the functions corresponding to $(\lambda,m_\lambda)$ and $(\lambda,M_\lambda),$ respectively. 
Then $e^{\lambda t}\,h_m(X_{t})\,G_t^{-1}$ and $e^{\lambda t}\,h_M(X_{t})\,G_t^{-1}$
are martingales.
For $z$ with $m_{\lambda}\leq z\leq M_{\lambda},$
let $h$ be the function corresponding to $(\lambda,z).$
Because $h$ can be expressed as a linear combination of $h_m$ and $h_M,$
$e^{\lambda t}\,h(X_{t})\,G_t^{-1}$ is also a martingale. Thus, $(\lambda,z)$ is in $\mathcal{A}.$

We now show that 
$e^{\lambda t}\,h(X_{t})\,G_t^{-1}$ is not a martingale if at least one of 
$(\lambda,m_\lambda)$ and $(\lambda,M_\lambda)$ is not in $\mathcal{A}.$
Write $h=c\,h_m+(1-c)h_M$ for some constant $0\leq c\leq 1.$
Because $e^{\lambda t}\,h_m(X_{t})\,G_t^{-1}$ and
$e^{\lambda t}\,h_M(X_{t})\,G_t^{-1}$ 
are positive local martingales, they are supermartingales.
We have
\begin{equation*}
\begin{aligned}
h(\xi)&=c\,h_m(\xi)+(1-c)h_M(\xi) \\
&\geq 
c\,\mathbb{E}[e^{\lambda t}\,h_m(X_{t})\,G_t^{-1}]+(1-c)\,\mathbb{E}[e^{\lambda t}\,h_M(X_{t})\,G_t^{-1}]\\
&=\mathbb{E}[e^{\lambda t}\,h(X_{t})\,G_t^{-1}]\;.
\end{aligned}
\end{equation*}
If at least one of 
$(\lambda,m_\lambda)$ and $(\lambda,M_\lambda)$ is not in $\mathcal{A},$ the above inequality is strict.
This completes the proof.

\section{Proof of Proposition \ref{prop:monotone_martingality}}
\label{app:pf_monotone_martingality}

\begin{theorem}\label{thm:criterion_martingality}
Let $(\beta,\phi)$ be a solution pair of
$\mathcal{L}\phi=-\beta\phi$ with positive function $\phi.$
Then
$$e^{\beta t}\,\phi(X_{t})\,\phi^{-1}(\xi)\,G_t^{-1}$$
is a martingale if and only if the diffusion induced by $(\beta,\phi)$
$$dX_t=(k+\sigma^2\phi'\phi^{-1})(X_t)\,dt+\sigma(X_t)\,dB_t$$
does not explode. 
\end{theorem}
\begin{proof}
Set $$H_t:=\exp{\left(-\frac{1}{2}\int_{0}^t v^2(X_s)\,ds-\int_0^t v(X_s)\,dW_s\right)}\;.$$
Define a new measure $\mathbb{L}$ by setting the Radon-Nikodym derivative
$d\mathbb{L}= H_t\,d\mathbb{Q}.$
By the Girsanov theorem, we know that a process $Z_t$ defined by
$$dZ_t=v(X_t)\,dt+dW_t$$
is a Brownian motion under $\mathbb{L}.$
It follows that
\begin{equation*}
\begin{aligned}
dX_t
&=(b-v\sigma)(X_t)\,dt+\sigma(X_t)\,dZ_t\\
&=k(X_t)\,dt+\sigma(X_t)\,dZ_t\;.
\end{aligned}
\end{equation*}
We used that that $k(x)=(b-v\sigma)(x).$
We show that 
$e^{\beta t}\,\phi(X_{t})\,G_t^{-1}$ is a martingale under $\mathbb{Q}$ if and only if
$e^{\beta t-\int_0^tr(X_s)\,ds}\,\phi(X_{t})$
is a martingale under $\mathbb{L}.$
It is because
\begin{equation*}
\begin{aligned}
\mathbb{E}^{\mathbb{L}}[e^{\beta t-\int_0^tr(X_s)\,ds}\,\phi(X_{t})]
&=\mathbb{E}^{\mathbb{Q}}[e^{\beta t-\int_0^t r(X_s)\,ds}\,\phi(X_{t})\,H_t]\\
&=\mathbb{E}^{\mathbb{Q}}[e^{\beta t}\,\phi(X_{t})\,G_t^{-1}]\;.
\end{aligned}
\end{equation*}
Since both are supermatingales, this computation gives the desired result.

From the contents of \cite{Pinsky} on page 212 and 215, we know that
$e^{\beta t-\int_0^tr(X_s)\,ds}\,\phi(X_{t})$
is a martingale under $\mathbb{L}$ if and only if
$$dX_t=(k+\sigma^2\phi'\phi^{-1})(X_t)\,dt+\sigma(X_t)\,dZ_t$$
does not explode. This completes the proof.
\end{proof}

\begin{lemma}
Assume $\delta<\lambda\leq\overline{\beta}.$
Let $g$ and $h$ be the functions corresponding to tuple $(\delta,M_\delta)$ and $(\lambda,M_\lambda),$ respectively. Then we have $g'g^{-1}>h'h^{-1}.$ 
\end{lemma}
\begin{proof}
For convenience, 
we may assume that $\xi=0$, the left boundary is $-\infty$ and the right boundary is  $\infty.$
Recall $q(x)=e^{-\int_{0}^{x}\frac{k(y)}{\sigma^{2}(y)}\,dy}$ in Proposition \ref{prop:ODE}.
Write $h=uq$ and $g=vq.$
Define $\Gamma=\frac{h'}{h}-\frac{g'}{g}=\frac{u'}{u}-\frac{v'}{v}.$
Then 
$$\Gamma'=-\Gamma^{2}-\frac{2v'}{v}\Gamma-\frac{2(\lambda-\delta)}{\sigma^{2}}\;.$$
It suffices to show that $\Gamma(x)<0$ for all $x.$
First, we show this for $x>0.$
We know $\Gamma(0)=M_\lambda-M_\delta<0.$ We have that $\Gamma(x)<0$ for all $x>0$ because $\Gamma$ ever gets close to $0,$ then the term $-\frac{2(\lambda-\delta)}{\sigma^{2}}$ dominates the right hand side of the equation above.

We now show that $\Gamma(x)<0$ for all $x<0.$
Suppose there exists $x_1<0$ such that $\Gamma(x_1)\geq0.$
Then for all $z<x_1,$ it is obtained that $\Gamma(z)>0$
since if $\Gamma$ ever gets close to $0,$ then term $-\frac{2(\beta-\lambda)}{\sigma^{2}}$ dominates the right hand side of the equation.
Fix a point $x_0$ such that $x_0<x_1,$ thus $\Gamma(x_0)>0.$  
For $z<x_0,$ we have
$$-\frac{2v'(z)}{v(z)}=\frac{\Gamma'(z)}{\Gamma(z)}+\Gamma(z)+\frac{2(\beta-\lambda)}
{\sigma^{2}(z)}\cdot\frac{1}{\Gamma(z)}\;.$$
By the same argument in Appendix \ref{app:pf_main_thm}, we have
\begin{equation*}
\begin{aligned}
\int_{-\infty}^{x_{0}}\frac{1}{v^{2}(y)}\,dy
< \infty \;.
\end{aligned}
\end{equation*}
This implies that
the diffusion process induced by tuple $(\delta,M_{\delta})$
is attracted to the left boundary, which is a contradiction.
\end{proof}

We now prove Proposition \ref{prop:monotone_martingality}.

\begin{proof}
Assume $\delta<\lambda\leq\overline{\beta}.$
Let $g$ and $h$ be the functions corresponding to tuple $(\delta,M_\delta)$ and $(\lambda,M_\lambda),$ respectively. 
Suppose that $e^{\delta t}\,g(X_{t})\,G_t^{-1}$ is a martingale. 
Then by Proposition \ref{prop:dynamics_under_P} and Theorem \ref{thm:criterion_martingality}, we know that
the process $X_t$ under the transformed measure with respect to $(\delta,M_\delta)$ satisfies
$$dX_t=(k+\sigma^2 g'g^{-1})(X_t)\,dt+\sigma(X_t)\,dB_t$$
and does not explode to $\infty.$ 
The above lemma says $g'g^{-1}>h'h^{-1}.$
By the comparison theorem, we know that
a process $Y_t$ with 
$$dY_t=(k+\sigma^2 h'h^{-1})(Y_t)\,dt+\sigma(Y_t)\,dB_t$$
satisfies $Y_t\leq X_t$ almost everywhere, thus $Y_t$
does not explode to $\infty.$ On the other hand, it is clear that $Y_t$ does not explode to $-\infty$ because it is non-attracted to the left boundary.
By Theorem \ref{thm:criterion_martingality}, we conclude that
$e^{\lambda t}\,h(X_{t})\,G_t^{-1}$ is a martingale. 
\end{proof}

\section{An example for Corollary \ref{cor:two_cases}}
\label{app:example_2nd}

Consider the equation $\mathcal{L}h=-\lambda h.$
Recall that
$\overline{\beta}:=\max\{\,\lambda\,|\,(\lambda,z)\in\mathcal{A}\,\}.$
That is, $\overline{\beta}$ is the maximum value among all the $\lambda$'s of the solution pair $(\lambda,h)$ with $h(\xi)=1$ and $h(\cdot)>0.$
In this section, we explore an example such that 
$\mathcal{L}h=-\overline{\beta} h$
has two linearly independent positive solutions.

Let $\mathcal{L}$ be such that
$$\mathcal{L}h(x):=h''(x)+\frac{x}{(1+x^{2})^{3/4}}\,h'(x)\quad \text{ for } \; x\in\mathbb{R}\;.$$
First, $\mathcal{L}h(x)=0$ has two linearly independent positive solutions:
$h_{1}(x)\equiv 1$ and $h_{2}(x)=\int_{-\infty}^{x}e^{-\int_{0}^{y}u(z)\,dz}\,dy$ where $u(z)=\frac{z}{(1+z^{2})^{3/4}}.$

It is enough to show that $\overline{\beta}=0.$
That is, for any fixed $\lambda>0,$ the equation $\mathcal{L}h=-\lambda h$ has no positive solutions.
Suppose there exists such a positive solution $h.$
Define a sequence of functions by $g_{n}(x):=\frac{h(x+n)}{h(n)}$ for $n\in\mathbb{N}.$
By direct calculation, $g_{n}$ satisfies the following equation:
\begin{equation} \label{eqn:g_n}
g_{n}''(x)+\frac{x+n}{(1+(x+n)^{2})^{3/4}}\,g_{n}'(x)=-\lambda g_{n}(x)\;.
\end{equation}
By the Harnack inequality stated below, we have that
$(g_{n})_{n=1}^{\infty}$ is equicontinuous on each compact set on $\mathbb{R};$
thus we can obtain a subsequence $(g_{n_k})_{k=1}^{\infty}$ such that the subsequence converges on $\mathbb{R},$ say the limit function $g.$
Since $g_{n}$ is positive, the limit function $g$ is nonnegative and $g$ is a nonzero function because $g(0)=\lim_{k\rightarrow\infty} g_{n_k}(0)=1.$ 
On the other hand, it can be easily shown that the limit function $g$ satisfies 
$g''(x)=-\lambda g(x)$
by taking limit $n\rightarrow\infty$ in equation \ref{eqn:g_n}.
Clearly there does not exist a nonzero nonnegative solution of this equation when $\lambda>0.$
This is a contradiction.
The author appreciates Srinivasa Varadhan for this example.

\begin{theorem} \textnormal{(Harnack inequality)}

\noindent Let $h:\mathbb{R}\rightarrow\mathbb{R}$ be a positive solution of 
$$a(x)h''(x)+b(x)h'(x)+c(x)h(x)=0\;.$$
Assume that $a(x)$ is bounded away from zero; that is, there is a positive number $l$ such that $a(x)\geq l>0.$ 
Suppose that $a(x),|b(x)|$ and $|c(x)|$ are bounded by a constant $K.$
Then for any $z>0,$
there exists a positive number $M=M(z,K)$ (depending on $z$ and $K,$ but on neither $a(\cdot),b(\cdot),c(\cdot)$ nor $h(\cdot)$)
such that
$\frac{h(x)}{h(y)}\leq M$
whenever $|x-y|\leq z.$
\end{theorem}

\section{Reference Functions}
\label{sec:ref_ftn}

In this section, we focus on the function $\phi$
rather than the value $\beta.$
We assume that we roughly know the behavior of $\phi;$
for example, we know a function $f$ such that $\phi^{-1}f$ is bounded below and above
or such that $\mathbb{E}^{\mathbb{P}}_{\xi}
\left[ (\phi^{-1}f)(X_{t})
\right]$ converges to a nonzero constant.
Knowing $f$ means that we have information
about $\phi$ near the area where the process $X_{t}$ lies with high probability under the
objective measure.
Such a function $f$ is called a {\em reference function} of $\phi.$
More generally and more formally, we define a reference function in the following way.
\begin{definition}
A positive function $f$ is called a {\em reference function} of $\phi$ if
$$\lim_{t\rightarrow\infty}\frac{1}{t}\cdot\log\mathbb{E}^{\mathbb{P}}_{\xi}
\left[ (\phi^{-1}f)(X_{t})
\right]=0\;.$$
\end{definition} 

\begin{proposition} \label{prop:ref_ftn_equi}
Knowing a reference function is equivalent to knowing the value $\beta.$
\end{proposition}

\begin{proof}
Suppose we know a reference function $f$ of $\phi.$
From
\begin{equation*}
\begin{aligned}
\mathbb{E}^{\mathbb{Q}}_{\xi}[G_t^{-1}f(X_{t})]
=\mathbb{E}^{\mathbb{P}}_{\xi}\left[(\phi^{-1}f)(X_{t})\right]\phi(\xi)\,e^{-\beta t} \;,
\end{aligned}
\end{equation*}
and by the definition of the reference function, we have that
\begin{equation*}
\begin{aligned}
\lim_{t\rightarrow\infty}\frac{1}{t}\cdot\log \mathbb{E}^{\mathbb{Q}}_{\xi}[ G_t^{-1}f(X_{t})]
=-\beta\;.
\end{aligned}
\end{equation*}
Hence, we know the value $\beta.$
Conversely, suppose we know the value $\beta.$
We show that for any admissible pair $(\beta,f)$ of $\mathcal{L}f=-\beta f,$ $f$ is a reference function.
We have
$$\mathbb{E}^{\mathbb{Q}}_{\xi}[G_t^{-1}\,f(X_{t})]=e^{-\beta t}f(\xi)$$
and thus
$$\mathbb{E}^{\mathbb{P}}_{\xi}[(\phi^{-1}f)(X_{t})]=\mathbb{E}^{\mathbb{Q}}_{\xi}
[G_t^{-1}f(X_{t})]\,e^{\beta t} \,\phi^{-1}(\xi)=(\phi^{-1}f)(\xi).$$
This completes the proof.
\end{proof}

\end{document}